\documentclass[a4paper,11pt]{article}

\usepackage{amsthm,amsmath,amssymb,amsfonts,pstricks,pst-node,graphics}

\DeclareMathAccent{\wtilde}{\mathord}{largesymbols}{"65}
\DeclareMathAccent{\what}{\mathord}{largesymbols}{"62}

\newcommand\cS{{\mathcal S}}

\newcommand\Z{{\mathbb Z}}

\newcommand\bbbr{{\mathbb R}}

\newcommand\cH{{\mathcal H}}

\newcommand\cL{{\mathcal L}}

\newcommand\cA{{\mathcal A}}
\newcommand\cB{{\mathcal B}}
\newcommand\cU{{\mathcal U}}
\def\bbbc{{\mathbb C}}

\renewcommand\i{{\iota}}

\newcommand\id{\mbox{\rm id}}
\def\acS{{\, \xrightarrow{\cS}\, }}

\newcommand\bx{{\bf x}}
\newcommand\by{{\bf y}}
\newcommand\bX{{\bf X}}
\newcommand\bY{{\bf Y}}

\newcommand\bz{{\bf z}}
\newcommand\bZ{{\bf Z}}
\newcommand\bF{{\bf F}}
\newcommand\bG{{\bf G}}
\newcommand\bbf{{\bf f}}
\newcommand\bbg{{\bf g}}

\newcommand\bsu{{\boldsymbol u}}
\newcommand\bsv{{\boldsymbol v}}

\newcommand\bsalpha{{\boldsymbol \alpha}}

\newtheorem{Def}{Definition}
\newtheorem{The}{Theorem}
\newtheorem{Pro}{Proposition}
\newtheorem{Lem}{Lemma}

\newtheorem{Cor}{Corollary}

\makeatletter
\def\wb{\accentset{{\cc@style\underline{\mskip10mu}}}}
\makeatother

\begin{document}
\title{Darboux transformation for the vector sine-Gordon equation and 
integrable equations on a sphere}
\author{ Alexander V. Mikhailov$^{\star}$, Georgios Papamikos$^{\dagger}$ and Jing Ping Wang$ ^\dagger $
\\
$\dagger$ School of Mathematics, Statistics \& Actuarial Science, University of Kent, Canterbury, UK \\
$\star$ School of Mathematics, University of Leeds, Leeds, UK
}
\date{}
\maketitle

\begin{abstract}
We propose a method for construction of  Darboux transformations, which is a 
new development of the dressing method for 
 Lax operators invariant under a reduction group. We apply the method to
the vector sine-Gordon equation and derive its B{\"a}cklund transformations.  We
show that there is a new Lax 
operator  canonically associated with our Darboux transformation resulting 
an evolutionary 
differential-difference system on a sphere. The latter is a generalised symmetry 
for the chain of B\"acklund transformations. Using the 
re-factorisation approach and the Bianchi permutability of the Darboux 
transformations we derive new vector Yang-Baxter map and integrable discrete 
vector sine-Gordon equation on a sphere.
\end{abstract}

\section{Introduction}

 Lax integrable partial differential equations have natural 
connections with integrable differential-difference and discrete equations via 
Darboux transformations. Namely, given a Lax representation for a partial 
differential equation, we can systematically construct Darboux transformations 
whose Bianchi permutability condition leads to an integrable difference 
equation, while the corresponding B\"acklund transformations are (often 
nonlocal) symmetries of these difference equations and are integrable 
differential--difference equations in their own right.

In this paper we make steps towards the development of a systematic approach to 
construction of Darboux transformations, associated difference systems 
and integrable maps based on a natural extension of the dressing method. 
The motivation for this line of research is to cooperate the reduction groups 
\cite{mik79,mik80,mik81} of Lax representations of
integrable partial differential equations to integrable 
difference  equations 
and corresponding Darboux transformations. The aim of our project is
to describe all elementary Darboux transformations for Lax operators on  
Kac-Moody
algebras
and automorphic Lie algebras with finite reduction groups. Recently, 
the authors of \cite{smx, krm13} have completed a comprehensive study for 
the Lax operators of the nonlinear Schr{\"o}dinger equation type and they 
derived some new discrete equations and new Yang-Baxter maps.
In the paper \cite{MPW} the method was applied to the case of Lax operators 
with  the reduction group generated by inner and outer automorphisms 
of the Lax representation for two dimensional Volterra chain.

In this paper we study differential-difference and partial difference equations 
associated with the vector sine-Gordon equation
\begin{equation}\label{vsG}
 D_t \left(\frac{D_x{\boldsymbol \alpha}}{\beta}\right)={\boldsymbol \alpha}, 
\quad 
\beta^2+|\boldsymbol \alpha|^2=1,
\end{equation}
where dependent variables ${\boldsymbol \alpha}=(\alpha^1, \cdots, \alpha^n)^T$ 
is $n$-dimensional 
real vector and $\beta\in\bbbr$. Here and in what follows the upper index
$T$ denotes the transposition of a vector or a matrix and we denote 
${\boldsymbol
\alpha}^T {\boldsymbol \alpha}=|{\boldsymbol \alpha}|^2$. Sometimes it is 
convenient to write $\beta$ and ${\boldsymbol \alpha}$
as an $n+1$-dimensional unit vector ${\boldsymbol v}^T= 
\left(\begin{array}{cc}\beta, & {\boldsymbol \alpha}^T \end{array}\right) $.
Thus the vector sine-Gordon equation is an integrable partial 
differential system of equations with the vector dependent variable $\bsv$ is 
on a sphere  
$\bsv\in{ S}^{n}=\{\boldsymbol{x}\in\bbbr^{n+1}\, |\, 
 \boldsymbol{x}^T\boldsymbol{x}=1\}$.
 
Equation (\ref{vsG}) has a long history, it  first appeared in \cite{PoR79} 
viewed as a reduction of
the two-dimensional  $\sigma$-model \cite{mr80c:81115}.  Its 
Lax representation was
given in \cite{EP79} and its Lagrangian formulation in \cite{bps95}. Later,
this equation reappeared in the study of moving frames in the Riemannian 
geometry \cite{wang02}. The dressing method for construction of multi-soliton 
solutions of (\ref{vsG}) and study of soliton interactions has been recently 
developed in \cite{MPW15}. 

In this paper we use the Lax representation 
(\ref{lax}), (\ref{JUV}) proposed in \cite{wang02} and studied in \cite{MPW15}.
For these Lax operators we derive a Darboux transformation (Theorem 
\ref{th1}) with Darboux matrix
\begin{equation}\label{darm}
 M_\nu(\bsu)=I_{n+2}+\frac{ i \nu}{\lambda-i\nu}Q a\rangle \langle a
-\frac{ i \nu}{\lambda+i\nu} a\rangle \langle a Q,\quad
\langle a=(i,\bsu^T),\quad \bsu=
\frac{{\boldsymbol v}_1+{\boldsymbol v}}{|{\boldsymbol
v}_1+{\boldsymbol v}|},
\end{equation}
where  $I_{n+2}$ is the identity matrix of size $n+2$, 
$Q=\mbox{diag}(-1,1,\ldots,1)$, and 
rigorously prove that 
\begin{eqnarray}
&&D_t\left(\frac{{\boldsymbol v}_1+{\boldsymbol v}}
{|{\boldsymbol v}_1+{\boldsymbol v}|}\right)
=-\frac{1}{2 \nu}({\boldsymbol v}_1-{\boldsymbol v});
\label{1}\\
&&\frac{D_x \boldsymbol \alpha_1}{\beta_1}-\frac{D_x\boldsymbol 
\alpha}{\beta}=-\frac{2 \nu}{|{\boldsymbol
v}_1+{\boldsymbol v}|}\left({\boldsymbol 
\alpha}_1+{\boldsymbol \alpha} \right),\label{2}
\end{eqnarray}
 is the B\"acklund transformation which depends on an arbitrary constant 
parameter  $0\ne \nu\in \mathbb{R}$ and relates two solutions
$\bsv$ 
and $\bsv_1$ of the vector sine-Gordon equation (\ref{vsG}). Equations 
(\ref{1}) and (\ref{2}) can also be seen as two non-evolutionary 
differential--difference equations on a lattice with variables 
$\bsv_j\in S^n,\ j\in\Z$ and with a shift operator $\cS_\nu: \bsv_{j}\mapsto 
\bsv_{j+1}$, assuming $\bsv_0=\bsv$.

For the Darboux matrix (\ref{darm}) there exists a new Lax operator
\begin{equation}\label{cB}
 \cB_\nu=D_{\tau }-i\nu\frac{Qa_{-1}\rangle\langle 
a-a\rangle\langle a_{-1}Q}{(\lambda-i\nu)\langle 
aQa_{-1}\rangle}+i\nu\frac{a_{-1}\rangle\langle 
aQ-Qa\rangle\langle a_{-1}}{(\lambda+i\nu)\langle aQa_{-1}\rangle}, \qquad
a_{-1}\rangle=\cS_\nu^{-1} (a)\rangle
\end{equation}
canonically associated with it. The existence of 
a Lax operator which is canonically associated with a Darboux matrix is quite 
remarkable. We first observed it in \cite{MPW} where a local generalised 
symmetry for non-evolutionary  differential--difference equations (B\"acklund 
transformations for two--dimensional Volterra chain) has been found. Our 
analysis of the variety of Darboux--Lax integrable differential--difference 
systems shows that local symmetries often correspond to such operators. They 
deserve more attention and we are planning to develop this direction of 
research in future publications. 

The compatibility of the Darboux  matrix 
(\ref{darm}) and Lax operator (\ref{cB}) also results in an evolutionary 
differential--difference equation
\begin{eqnarray}\label{3}
D_{\tau }(\boldsymbol v) =\frac{|\boldsymbol v_{-1}+\boldsymbol v|^2 
(\boldsymbol 
v_1+\boldsymbol v)-|\boldsymbol v_1+\boldsymbol v|^2
(\boldsymbol v_{-1}+\boldsymbol v)}{( \boldsymbol v_1^T+\boldsymbol v^T)( 
\boldsymbol v_{-1}+\boldsymbol v )+|\boldsymbol
v_{-1}+\boldsymbol v||\boldsymbol v_1+\boldsymbol v|},\qquad \bsv_k =\cS_\nu^k 
(\bsv),
\end{eqnarray}
which is a local symmetry of the nonevolutionary integrable systems (\ref{1}) 
and (\ref{2}). System (\ref{3}) is known, and it  was  found by Adler 
\cite{Adler08a}  in his classification of isotropic 
integrable Volterra-type lattices on the sphere with generalised symmetries. 
In this paper we equip this system with a Lax-Darboux representation 
and connect it  to the vector sine-Gordon equation (\ref{vsG}) and its
B\"acklund chains (\ref{1}) and (\ref{2}).

The Bianchi permutability condition for two Darboux transformations with
distinct parameters $\mu\ne\pm\nu$ resulting in two shift operators 
$\cS_\nu,\ \cS_\mu$  leads to
the integrable discrete equation
\begin{equation}\label{diffeq0}
 \bsv_{1,0}=-\bsv+2\bX (\bX^T,\bsv)
\end{equation}
where  $\cS_{\nu}^n\cS_{\mu}^m \bsv=\bsv_{n,m}$ and   
\begin{equation*}
\bX =\frac{(\nu^2-\mu^2) \bx+2\mu(\nu+\mu (\bx^T\by))\by}{\nu^2+\mu^2+2\mu\nu 
(\bx^T\by)} , \qquad
\bx=\frac{\bsv_{1,1}+\bsv_{0,1}}{|\bsv_{1,1}+\bsv_{0,1}|},\quad  
\by=\frac{\bsv_{0,1}+\bsv}{|\bsv_{0,1}+\bsv|}.
 \end{equation*}  
 Equation (\ref{diffeq0}) can also be uniquely resolved with respect to the 
variable $\bsv_{0,1}$.  B\"acklund transformations (\ref{1}), (\ref{2}) and 
evolutionary equation 
(\ref{3}) are its non-local and local symmetries.

Moreover, the re-factorisation of the product of two Darboux matrices 
$M_\nu(  \bx)M_\mu(  \by)=M_\mu(  \bY)M_\nu(  
\bX)$  leads to a new 
involutive Yang-Baxter map $R_{\nu,\mu}:S^n\times S^n\mapsto S^n\times S^n$
\begin{equation}\begin{split}
 R_{\nu,\mu}: \left(\begin{array}{c}
              \bx\\ \by
              \end{array}\right)\ \mapsto\ \left(\begin{array}{c}
              \bX
              \\ \bY
              \end{array}\right)=\left(\begin{array}{c}
              \dfrac{(\nu^2-\mu^2) \bx+2\mu (\nu+\mu 
(\bx^T\by))\by}{\nu^2+\mu^2+2\mu\nu (\bx^T\by)}
              \\ \\ \dfrac{(\mu^2-\nu^2) \by+2\nu (\mu+\nu 
(\bx^T\by))\bx}{\nu^2+\mu^2+2\mu\nu (\bx^T\by)}
              \end{array}\right) 
                \end{split},\qquad \bx,\by,\bX,\bY\in { S}^n.
\end{equation}

The arrangement of this paper is as follows: In Section \ref{sec2}, we give 
basic 
definitions related to the Lax-
Darboux scheme such as Darboux transformation, the shift operators and
B\"acklund 
transformation.
In Section \ref{sec3}, we derive the B\"acklund transformations for 
the 
vector sine-Gordon equation (\ref{vsG}) together with the Darboux matrix. In
Section \ref{sec4}, we construct a local generalised symmetry for these
non-evolutionary equations using the Darboux matrix.  Meanwhile, we provide the
Lax representation for this symmetry flow, which is one of isotropic integrable
Volterra-type lattices on the sphere \cite{Adler08a}. In Section \ref{sec5}, we
use the re-factorisation of two Darboux matrices to derive the corresponding
Yang-Baxter map and integrable discrete vector sine-Gordon equation, which are
new to the best of our knowledge.

\section{The Lax-Darboux scheme}\label{sec2}
In this section, we recall the Lax representation of the vector sine-Gordon 
equation and 
introduce some basic definitions such as Darboux
transformation and 
B\"acklund transformation.

The vector sine-Gordon equation (\ref{vsG}) is equivalent to the compatibility 
condition $ [\cL,\cA]=0$
for two linear problems \cite{wang02}, \cite{MPW15}
\begin{equation}\label{linprobs}
  \cL \Psi=0,\qquad \cA \Psi=0, 
\end{equation}
where
\begin{eqnarray}\label{lax}
 \cL=D_x-\lambda J -U\quad \mbox{and} \quad \cA=D_t+\lambda^{-1} V,
\end{eqnarray}
and
\begin{eqnarray}\label{JUV}
J=\left(\begin{array}{ccc} 0 & 1 &{\bf 0}^T\\-1 & 0 & {\bf 0}^T\\{\bf 0} &{\bf 
0}&0_n \end{array}\right),\
U=\left(\begin{array}{ccc} 0 & 0 &{\bf 0}^T\\0 & 0 & -{\boldsymbol 
\alpha}_x^T/\beta\\{\bf 0}
&
{\boldsymbol \alpha}_x/\beta&0_n \end{array}\right),\ 
V=\left(\begin{array}{cc} 0 & {\boldsymbol v}^T\\-{\boldsymbol v}
&0_{n+1} \end{array}\right)\, ,
\end{eqnarray}
where ${\bf 0}$ is $n$-dimensional zero column vector
and $0_k$ denotes the $k\times k$ zero matrix.
Without causing confusion, we sometimes simply write $0$ instead for either of
them. 

A Darboux transformation is a linear map acting on a fundamental solution 
\begin{equation}\label{SPsi}
\Psi\mapsto \overline{\Psi}=M\Psi,\qquad \det\,M\ne 0
\end{equation}
such that the matrix function $\overline{\Psi}$ is a fundamental solution of 
the linear problems 
\begin{equation}\label{lax1}
 \overline{\cL}\, \overline{\Psi}=0,\qquad \overline{\cA}\, \overline{\Psi}=0
\end{equation}
with new ``potentials''  $\overline{{\boldsymbol \alpha}}$ and 
$\overline{\beta}$.
The matrix $M$ is often called the Darboux matrix.
From the compatibility of (\ref{SPsi}) and (\ref{lax1}) it follows that 
\begin{eqnarray}
&&D_x(M)=(\lambda J+\overline{U})M-M(\lambda J+U)=\lambda [J,M]+\overline{U}
M-MU;\label{darx}\\
&&D_t(M)=(-\lambda^{-1}\overline{V})M-M(-\lambda^{-1}V)=-\lambda^{-1}(\overline{
V}
M-MV),
\label{dart}
\end{eqnarray}
Equations lead  to {\sl B\"acklund transformations} for the vector sine-Gordon
equation, which relate  two solutions 
${\boldsymbol \alpha},\beta$
and $\overline{{\boldsymbol \alpha}},\overline{\beta}$ of  (\ref{vsG}).

A Darboux transformation maps one compatible system (\ref{linprobs}) into 
another one (\ref{lax1}). It defines a map $\cS: \boldsymbol v
\mapsto\overline{\boldsymbol v}$.
The map (\ref{SPsi}) is invertible ($\det\,M\ne 0$) and it can be iterated
\[
  \cdots\underline{\Psi}\acS \Psi\acS\overline{\Psi}
  \acS\overline{\overline{\Psi}}\acS\cdots\, .
\]
We introduce notations
$$
 \ldots \Psi_{-1}=\underline{\Psi},\ \Psi_0=\Psi,\ \Psi_1=\overline{\Psi},\ 
\Psi_2=\overline{\overline{\Psi}}, \ldots,
$$
$$ \ldots {\boldsymbol v}_{-1}=\underline{{\boldsymbol v}},\ {\boldsymbol 
v}_0={\boldsymbol v},\ {\boldsymbol v}_1=\overline{{\boldsymbol v}},\ 
{\boldsymbol v}_2=\overline{\overline{{\boldsymbol v}}}, \ldots\, .$$
 In these notations
the maps $\cS$ and $\cS^{-1}$ increase and decrease the subscript index 
by one, and therefore we shall call it  a $\cS$--shift, or shift operator $\cS$.
The resulting B\"acklund transformations from the Lax-Darboux representations 
(\ref{darx}) and (\ref{dart}) are integrable differential difference equations.

A Darboux transformation with a parameter $\mu$ denoted by $M_{\mu}$ results in 
the $\cS_{\mu}$ shift. If we also consider a Darboux transformation  
with a different choice of the parameter $\nu$, then the corresponding shift we
denote $\cS_{\nu}$. 
Commuting shifts act on ${\mathbb Z}^2$ lattice where with the vertex $(n,m)$ we
associate the variable $\bsv_{n,m}=\cS^n_{\nu}
\cS^m_{\mu}\bsv$. 
Commutativity of the shifts is (Bianchi permutability) 
\begin{eqnarray}\label{dis}
\cS_{\nu}(M_\mu)M_\nu-\cS_{\mu}(M_\nu)M_\mu=0 
\end{eqnarray}
is equivalent to a 
system of partial-difference equations and we call it Darboux representation 
for  this system of partial-difference equations. In literature such 
representation and Darboux matrices sometimes referred as discrete Lax 
representation and discrete Lax operators respectively.
Differential difference equations (\ref{darx}) and (\ref{dart}) are the
symmetries of this partial-difference equation. 

We can also consider the re-factorisation of a
product of two Darboux matrices, which leads to the Yang-Baxter map.
In next sections, we construct the Darboux matrices for the vector
sine-Gordon equation (\ref{vsG}), and further produce its B\"acklund
transformations, the associated Yang-Baxter map and the integrable difference
equation.

\section{Invariant Darboux matrix under the group $\mathbb{Z}_2\times 
\mathbb{Z}_2\times
\mathbb{Z}_2$}\label{sec3}

The Lax operators $\mathcal{L}$ and $\mathcal{A}$ are invariant under the
reduction group 
$\mathbb{Z}_2\times \mathbb{Z}_2\times
\mathbb{Z}_2$ generated by three transformations $\i$ $r$ and $s$ satisfying 
$\i^2=r^2=s^2={\rm id}$.
In particular, for operator $\cL$ these transformations are:
\begin{eqnarray}
&&\i:\mathcal{L}(\lambda)\rightarrow
-\mathcal{L}^{\dagger}(\lambda)=\mathcal{L}(\lambda),
\label{redgroup1}\\
&&r:\mathcal{L}(\lambda)\rightarrow {
\mathcal{L}^*({\lambda}^*)}=\mathcal{L}(\lambda),
\label{redgroup3}\\
&&s:\mathcal{L}(\lambda)\rightarrow
Q\mathcal{L}(-\lambda)Q=\mathcal{L}(\lambda),
\label{redgroup2}
\end{eqnarray}
where $\mathcal{L}^{\dagger}(\lambda)$ is the formally 
adjoint operator defined by $\mathcal{L}^{\dagger}(\lambda)=-D_x-\lambda
J^T-U^T$, matrix $
Q=\mbox{diag}(-1,1,\ldots,1)$  and $^*$ denotes the complex conjugation. 
 
We assume that Darboux matrix $M(\lambda)$ is a rational function of the 
spectral parameter $\lambda$ and is invariant with respect to the reduction 
group (\ref{redgroup1})--(\ref{redgroup2}) action, namely,
\begin{eqnarray}\label{invMi}
&& \i:M(\lambda)\mapsto(M(\lambda)^{-1})^{T}=M(\lambda), \\ \label{invMr}
&& r:M(\lambda)\mapsto {M^*({\lambda}^*)}=M(\lambda), \\ \label{invMs}
&& s:M(\lambda)\mapsto Q M(-\lambda)Q=M(\lambda).
\end{eqnarray}
Notice that the action of the automorphism $\i$ (\ref{invMr}) is different from 
(\ref{redgroup1})
since $M$ is an element of a Lie group 
rather than Lie algebra.

It is easy to show (see Proposition 1 in \cite{MPW15}) that a 
$\lambda$--independent Darboux matrix results in a constant linear map.

\begin{Pro}{\rm \cite{MPW15}} \label{prop1}  
Assume Darboux matrix $M$ is invariant under the group $\mathbb{Z}_2\times 
\mathbb{Z}_2\times
\mathbb{Z}_2$ and it is independent of spectral parameter $\lambda$.  Then $M$ 
is a constant matrix and of the form
\begin{equation}\label{M0}
 M=\pm \left(\begin{array}{ccc}
                 1&0&{\bf 0}^T\\
                 0&1&{\bf 0}^T\\
                 {\bf 0}&{\bf 0}&\Omega
                \end{array}\right),
\end{equation}
where ${\bf 0}$ is $n$-dimensional zero column vector and a constant ($x,t$ - 
independent) matrix $\Omega\in O(n,\bbbr)$. 
\end{Pro}
This implies that the constant Darboux matrix defines the transformation
\begin{eqnarray}\label{constr}
 \overline{\beta}=\beta \quad \mbox{and} \quad \overline{\boldsymbol \alpha} 
=\Omega {\boldsymbol \alpha}.
\end{eqnarray}
corresponding to $O(n,\bbbr)$ point symmetry of the vector sine-Gordon equation 
(\ref{vsG}).

A rational in $\lambda$ and reduction group invariant 
(\ref{invMi})--(\ref{invMs}) matrices $M$ must have poles on the orbits of the 
reduction group. They can be represented as a product of matrices which have 
poles in a single orbit. Here we shall study reduction group invariant Darboux 
matrix $M_\nu$ with two simple poles at $\lambda=\pm i\nu,\ \nu\ne 0,\ 
\nu\in\bbbr$. It  is the minimal possible number of simple poles for a 
reduction 
group invariant  Darboux matrix and it corresponds to the dressing matrix for a 
single kink solution \cite{MPW15}. 
It is of the following form (cf. Proposition
2 in \cite{MPW15}):
\begin{eqnarray}\label{Mmk0}
 M_\nu=I_{n+2}+\frac{ i \nu}{\lambda-i\nu} A
-\frac{ i \nu}{\lambda+i\nu} Q A Q, \qquad A\neq 0 .
\end{eqnarray}
\begin{Pro}\label{Aform}
If matrix (\ref{Mmk0}) is invariant under 
reduction group generated by $\i, r$ and $s$, it can be
represent as
\begin{eqnarray}\label{Mmk}
 M_\nu(\bsu)=I_{n+2}+\frac{ i \nu}{\lambda-i\nu} Q a\rangle \langle a
-\frac{ i \nu}{\lambda+i\nu} a\rangle \langle a Q,\qquad\langle
a=(i, \bsu^T), 
\end{eqnarray}
where $\bsu$ is a unit length real vector $\bsu\in S^n$.
\end{Pro}

\begin{proof} Note that the given matrix is invariant under the action of $s$.
It immediately follows from (\ref{invMi}) that $M_\nu M_\nu^T=I_{n+2}$, which
implies that
\begin{eqnarray}
&&A A^T=A^TA=0; \label{At}\\
&& (I_{n+2}-\frac{ 1}{2} QAQ)A^T+A(I_{n+2}-\frac{ 1}{2} QA^TQ) =0.\label{PA}
\end{eqnarray}
These identities correspond to the vanishing of the second and first order
poles at $\lambda=\pm i \nu$, respectively.
The invariance under the action of $r$ implies
\begin{eqnarray}
&&A^*=QAQ, \label{A*}
\end{eqnarray}
So the matrix $A$ satisfies the conditions of Lemma \ref{AA} in Appendix. 
Hence, it is of the form
$$A=h\left( \begin{array}{c} -i\\{\bf m}\end{array}\right) 
\left( \begin{array}{cc} i & {\bf u}\end{array}\right), \qquad
0\neq h\in \bbbr,\quad  {\bf m}, {\bf u}\in \bbbr^{n+1},\quad |{\bf m}|= |{\bf
u}|=1 .
$$
Substituting it into (\ref{PA}), we obtain 
$$h=1, \qquad {\bf m}={\bf u}.$$
Thus we have $A=Q a\rangle \langle a$, from which the result in the statement
immediately follows.
\end{proof}

This proposition confirms that the kink solutions of rank $1$ for the vector
sine-Gordon equation (\ref{vsG}) obtained in \cite{MPW15} are indeed generic
and there are no kink solutions of higher ranks. 

\begin{Pro}\label{prop4}
Matrix $M_\nu(\bsu)$ (\ref{Mmk}) satisfies the compatibility condition (\ref{dart})
if and only if
\begin{eqnarray}\label{eqdar}
&& \overline{V} M_\nu^0(\bsu)=M_\nu^0(\bsu) V,\label{conda}\\
&&D_t\left( Q a\rangle \langle a \right)+\frac{1}{i\nu}
 \overline{V} Q a\rangle \langle a-\frac{1}{i\nu} Q a\rangle
\langle a V=0,\label{eqtt}
\end{eqnarray}
where 
\[
 M_\nu^0(\bsu)=\lim_{\lambda\to 0}M_\nu(\bsu).
\]
\end{Pro}
\begin{proof} Taking the residue at $\lambda=0$ on both sides of
(\ref{dart}), we obtain (\ref{conda}). Taking the residue at $\lambda=i\nu$
on both sides of
(\ref{dart}), we obtain (\ref{eqtt}). The residue at $\lambda=-i \nu$ will
also vanishes due to the reduction group.
\end{proof}
We now convert the above conditions to the conditions for the components of 
vector $\langle a=(i, \bsu^T)$. 
\begin{Pro}
Matrix $M_\nu(\bsu)$ given by (\ref{Mmk})
satisfies (\ref{conda}) if
\begin{eqnarray}\label{res0}
 {\boldsymbol u}=\pm\frac{\overline{\boldsymbol
v}+{\boldsymbol v}}{|\overline{\boldsymbol v}+{\boldsymbol v}|}.
\end{eqnarray} 
\end{Pro}
\begin{proof} 
It follows from (\ref{Mmk}) that
\begin{eqnarray}\label{mo}
 M_\nu^0(\bsu)=I_{n+2}- Q a\rangle \langle a- a\rangle \langle a Q
=(M_\nu^0(\bsu))^T=(M_\nu^0(\bsu))^{-1}.
\end{eqnarray}
Thus we have $\overline{V} =M_\nu^0(\bsu) V M_\nu^0(\bsu)$. Using (\ref{JUV}), 
we 
get 
\begin{eqnarray*}
 &&
 \overline{\boldsymbol v}+{\boldsymbol v}=2  ({\boldsymbol u}^T {\boldsymbol v})
{\boldsymbol u}.
\end{eqnarray*}
It follows that $  ({\boldsymbol u}^T {\boldsymbol
v})^2=\frac{1}{2}(\overline{\boldsymbol v}^T{\boldsymbol v}+1)=\frac{1}{4} |
\overline{\boldsymbol v}+{\boldsymbol v}|^2.
$ Thus we have ${\boldsymbol u}=\pm\frac{\overline{\boldsymbol
v}+{\boldsymbol v}}{|\overline{\boldsymbol v}+{\boldsymbol v}|}$. 
\end{proof}

In the rest of the paper we assume that $\overline{\boldsymbol
v}+{\boldsymbol v}\ne {\bf 0}$ and choose the positive sign in (\ref{res0}).
We now compute the B\"acklund transformation and summarise the result for the 
Darboux matrix.

\begin{The}\label{th1}
The vector sine-Gordon equation (\ref{vsG}) possesses a B\"acklund transformation
\begin{eqnarray}
D_t\left(\frac{\overline{\boldsymbol v}+{\boldsymbol v}}
{|\overline{\boldsymbol v}+{\boldsymbol v}|}\right)
=-\frac{1}{2 \nu}(\overline{\boldsymbol 
v}-{\boldsymbol v}); \quad \boldsymbol v= 
\left(\begin{array}{c}\beta\\{\boldsymbol \alpha}
\end{array}\right),\label{alphaeqq} 
\end{eqnarray}
where  $\nu\in \mathbb{R}$ is constant.
The corresponding Darboux matrix is 
\begin{eqnarray}\label{darM}
 M_\nu(\bsu)=I_{n+2}+\frac{ i \nu}{\lambda-i\nu}Q a\rangle \langle a
-\frac{ i \nu}{\lambda+i\nu} a\rangle \langle a Q,\quad
\langle a= (i, {\boldsymbol u}^T),\quad {\boldsymbol u}=
\frac{\overline{\boldsymbol v}+{\boldsymbol v}}{|\overline{\boldsymbol
v}+{\boldsymbol v}|}.
\end{eqnarray}
\end{The}
\begin{proof}
We derive the B\"acklund transformation from (\ref{eqtt}), which is equivalent to
\begin{eqnarray}\label{eqv1}
 Q \left( a_t\rangle +\frac{1}{i\nu}
Q^{-1} \overline{V} Q a\rangle \right)\langle a + Q a\rangle\left( \langle a_t 
-\frac{1}{i\nu} 
\langle a V\right) =0.
\end{eqnarray}
Therefore, there exists a scalar $\gamma$ such that
\begin{eqnarray}\label{eqv2}
&& a_t\rangle +\frac{1}{i\nu} Q^{-1} \overline{V} Q a\rangle=a\rangle \gamma; 
\qquad
 a_t\rangle +\frac{1}{i\nu}   V a\rangle=-a\rangle \gamma. 
\end{eqnarray}
Taking the sum of these two equations we get 
\begin{eqnarray*}
 a_t\rangle +\frac{1}{2i\nu} \left( Q^{-1} \overline{V} Q +V\right) a\rangle=0.
\end{eqnarray*}
that is, \begin{eqnarray*}
&&({\boldsymbol v}^T-\overline{\boldsymbol
v}^T){\boldsymbol u}=0;\qquad  D_t{\boldsymbol u}=-\frac{1}{2 \nu}(\overline{\boldsymbol 
v}-{\boldsymbol v}).
\end{eqnarray*}
The first equation is satisfied due to (\ref{res0}) and $|\bsv|=|\bsv_1|=1$.
The second one is the required equation (\ref{alphaeqq}). It is easy to check
that equations (\ref{eqv1}) and (\ref{eqv2}) (with
$\gamma=\frac{1}{\nu}\bsv^T\bsu$) are satisfied due to (\ref{alphaeqq}).
\end{proof}

Using the compatibility condition (\ref{darx}) between the Darboux matrix 
(\ref{darM}) and matrix $U$, we get the following B\"acklund transformation 
with respect to $x$:
\begin{Cor}\label{cor1}
For the Darboux matrix  $M_\nu(\bsu)$ given in Theorem \ref{th1}, the 
compatibility condition (\ref{darx}) leads to the following \sl B\"acklund 
transformation (with respect to $x$):
\begin{eqnarray}\label{backx}
\frac{D_x(\overline{\boldsymbol \alpha)}}{\overline{\beta}}- 
\frac{D_x(\boldsymbol 
\alpha)}{\beta}=-\frac{2 \nu}{|\overline{\boldsymbol
v}+{\boldsymbol v}|}\left(\overline{\boldsymbol 
\alpha}+{\boldsymbol \alpha} \right).
\end{eqnarray}
\end{Cor}
\begin{proof}
Substituting the Darboux matrix $M_\nu(\bsu)$ into (\ref{darx}) and 
comparing the residues of all poles from both sides, we have
\begin{eqnarray}
&& \overline{U}=U+i \nu\left[ Q a\rangle \langle a-a\rangle \langle a Q , J 
\right];\label{darxx} \\
&& D_x\left( Q a\rangle \langle a \right)=\overline{U}  Q a\rangle \langle a - 
Q a\rangle \langle a U
+ i \nu \big[J,  Q a\rangle \langle a \big].\label{darxcom}
\end{eqnarray}
By direct calculation, we obtain the identity (\ref{backx}) in the statement 
from (\ref{darxx}). 

We now check that (\ref{darxcom}) satisfies under the condition of
(\ref{backx}). We rewrite it as
\begin{eqnarray}\label{eqvxcon}
 Q \left( D_x(a)\rangle -\overline{U} a\rangle +i \nu  J a\rangle  
\right)\langle 
a 
 + Q a\rangle\left(\langle D_x(a) +
\langle a U+i \nu \langle a J\right) =0.
\end{eqnarray}
Differentiating $|\overline{\boldsymbol
v}+{\boldsymbol v}|^2$ with
respect to $x$ and using $\beta_x=-\frac{1}{\beta}{\boldsymbol \alpha}^T 
{\boldsymbol
\alpha}_x$,  we have
\begin{eqnarray*}
&&
D_x(|\overline{\boldsymbol v}+{\boldsymbol v}|)
=\frac{1}{|\overline{\boldsymbol
v}+{\boldsymbol v}|} (\overline{\boldsymbol v}^T+{\boldsymbol
v}^T) D_x(\overline{\boldsymbol v}+{\boldsymbol 
v}) =\frac{1}{|\overline{\boldsymbol
v}+{\boldsymbol v}|}(\beta
\overline{\boldsymbol \alpha}^T-\overline{\beta} {\boldsymbol
\alpha}^T)\left(\frac{{\boldsymbol \alpha}_x}{\beta}-\frac{ 
\overline{\boldsymbol
\alpha}_x}{\overline{\beta}}\right)
\\
&&\qquad=\frac{2 \nu}{|\overline{\boldsymbol
v}+{\boldsymbol v}|^2}(\beta \overline{\boldsymbol
\alpha}^T-\overline{\beta} {\boldsymbol \alpha}^T)\left(\overline{\boldsymbol 
\alpha}+{\boldsymbol \alpha} \right) =\nu (\beta-\overline{\beta}) .
\end{eqnarray*}
Using it and (\ref{backx}), by direct calculation we obtain
\begin{eqnarray*}
&& \langle D_x(a) +\langle a U+i \nu \langle a J=-\frac{\nu 
(\overline{\beta}+\beta)
}{|\overline{\boldsymbol v}+{\boldsymbol v}|}\langle a;\\
&&D_x(a)\rangle -
 \overline{U} a\rangle +i \nu  J a\rangle=\frac{\nu (\overline{\beta}+\beta)
}{|\overline{\boldsymbol v}+{\boldsymbol v}|}a \rangle  .
\end{eqnarray*}
Substituting it into (\ref{eqvxcon}), we see the identity is valid and thus we
proved the statement.
\end{proof}

When $n=1$, we take $\alpha=\sin \theta$ and  $\beta=\cos \theta$. The
vector sine-Gordon equation (\ref{vsG})
reduces to the scalar sine-Gordon equation 
\begin{eqnarray}\label{ssg}
\theta_{xt}=\sin \theta 
\end{eqnarray}
We show that (\ref{alphaeqq}) and (\ref{backx}) give us its the well-known 
B\"acklund
transformation.

In this case, we have 
$\frac{ \overline{\bsv}+\bsv}{|\overline{\bsv}+\bsv|}=\left(\begin{array}{c} 
\cos \frac{\overline{\theta}+\theta}{2}\\
\sin \frac{\overline{\theta}+\theta}{2}\end{array}\right)$.
Substituting it into (\ref{alphaeqq}) and using trigonometric identities, we 
obtain
\begin{eqnarray}\label{sgst}
\overline{\theta}_t+\theta_t=\frac{2}{\nu}
\sin\tfrac{\theta-\overline{\theta}}{2}. \label{darsgt} 
\end{eqnarray}
Substituting the above $\alpha$ and $\beta$ into the B\"acklund
transformation (\ref{backx}), we get
\begin{eqnarray}\label{sgsx}
\overline{\theta}_x-\theta_x=-\frac{2\nu}{2 \cos
\tfrac{\overline{\theta}-\theta}{2}}(\sin \overline{\theta}+\sin \theta) = - 2
\nu \sin \tfrac{\overline{\theta}+\theta}{2}.
\end{eqnarray}

B\"acklund transformations (\ref{backx}), (\ref{sgsx}) provide us with
integrable discretisations of the vector (\ref{vsG}) and scalar (\ref{ssg})
sine-Gordon equations. There exists another integrable discretisation of
these equations, recently published in \cite{tsuchida15}, which has a different
nature and is not directly related to the B\"acklund transformations of the
vector
or scalar sine-Gordon equations.

\section{A local symmetry of the B\"acklund transformation}\label{sec4}
As discussed in Section \ref{sec2} the B{\"a}cklund transformations
(\ref{alphaeqq}) and (\ref{backx}) can 
be viewed as integrable differential--difference equations.
In this section, we show how to construct one of their generalised symmetries 
using the Darboux 
matrix (\ref{darM}). This symmetry is an integrable evolutionary equation, 
which belongs to the list of integrable Volterra-type lattice systems on a 
sphere satisfying the integrability conditions following from the existence of 
a formal recursion operator \cite{Adler08a}.  We construct a Lax 
operator, canonically associated with the Darboux matrix (\ref{darM}) which 
forms Lax-Darboux representation for this local symmetry.

The Lax operators $\cA, \cL$ (\ref{lax}) have simple poles at points of the 
degenerate orbits $\lambda=0, \infty$ and are invariant with respect to the 
reduction group generated by the transformations 
(\ref{redgroup1})-(\ref{redgroup2}).  Starting from the Darboux matrix
$M_\nu(\bsu)$ (\ref{darM}) we construct a new  Lax
operator 
\begin{equation}\label{laxtaunu}
 \cB_\nu=D_{\tau }-\cU_\nu,
\end{equation} 
which is invariant with respect to the reduction group and matrix 
$U_{\nu}$ has the same set of  poles as the logarithmic Fr\'echet derivative 
of the Darboux matrix. Thus, in the case of the Darboux matrix $M_\nu(\bsu)$ 
(\ref{darM}) matrix $\cU_\nu$ has the form
\begin{eqnarray}
&&\cU_\nu(\lambda)=  \frac{i\nu B}{\lambda
-i \nu}-\frac{i \nu Q B Q}{\lambda +i \nu}
\label{U}
\end{eqnarray}
and satisfies
\begin{eqnarray}
&& \cU_\nu(\lambda)=-\cU_\nu^{T}(\lambda), \qquad 
\cU_\nu^*(\lambda^*)=\cU_\nu(\lambda),  \qquad Q
\cU_\nu(-\lambda) Q^{-1}=\cU_\nu(\lambda). \label{InvU}
\end{eqnarray}
For  (\ref{U}) conditions (\ref{InvU}) are equivalent to
\begin{equation}\label{InvB}
 B=-B^T,\qquad B=QB^*Q.
\end{equation}
We recall that (see (\ref{darM}))
\begin{eqnarray}
&& M_\nu(\bsu)=I_{n+2}+\frac{i\nu 
}{\lambda-i\nu} A-\frac{i\nu }{\lambda+i\nu} QAQ, \qquad A=  Q 
a\rangle \langle a,\quad \langle a=(i,\bsu^T), \quad
|\bsu|=1.
\label{M}
\end{eqnarray}
The compatibility condition  
\begin{eqnarray}\label{zero}
D_{\tau } M_\nu(\bsu)-\cS_\nu (U_{\nu}) \ M_\nu(\bsu)+M_\nu(\bsu) \ 
U_{\nu}=0
\end{eqnarray}
for the Lax-Darboux pair
\begin{eqnarray}\label{lins}
 D_{\tau} 
\Psi=U_{\nu} \Psi, \qquad \cS_\nu \Psi=M_\nu(\bsu) \Psi 
\end{eqnarray}
 is equivalent to the following two equations
\begin{eqnarray}
&& B_1 A=A B, \label{con1}\\
&&A_{\tau }=B_1-B-\frac{1}{2} \left(B_1 Q AQ+Q B_1 QA-QAQB-AQBQ \right) ,
\label{con2}
\end{eqnarray}
where  $B_1=\cS_\nu(B)$.
Indeed, the left hand side of equation (\ref{zero}) is a rational 
matrix function of $\lambda$. This function is vanishing at 
$\lambda=\infty$ and conditions (\ref{con1}) and (\ref{con2}) are equivalent to 
the vanishing of second and first order poles at $\lambda=\pm i\nu$.

It is easy to check that 
$$B=\frac{\gamma}{\langle a Q a_{-1} \rangle}\left(Q a_{-1} \rangle \langle a- 
a \rangle \langle a_{-1} Q\right),\qquad a_{-1} \rangle=\cS_\nu^{-1}(a)\rangle
$$
is a solution of (\ref{con1}) satisfying the conditions (\ref{InvB}) if
$\gamma$ is any real function
of $\tau $. Without any loss of generality we can set $\gamma=1$ (by a point 
transformation). Substituting the expressions of $A$ and $B$ into
(\ref{con2}), we get
\begin{eqnarray*}
&&  Q \left( a_{\tau}+\frac{a_{-1} \rangle}{\langle a Q
a_{-1} \rangle}+\frac{\langle a_1  a\rangle}{2  \langle a_1 Q a
\rangle} Q a\rangle-\frac{ a_1 \rangle }{\langle a_1 Q a
\rangle}-\frac{\langle a a_{-1} \rangle}{2 \langle a Q a_{-1} \rangle}
Q a\rangle 
  \right)\langle a \\
&&\ + Q a\rangle\left(
\langle a_{\tau} -\frac{\langle
a_1}{\langle a_1 Q a\rangle} +\frac{\langle a_1  a\rangle}{2 \langle
a_1 Q a \rangle} \langle a Q+\frac{\langle a_{-1}}{ \langle a Q a_{-1}
\rangle}-\frac{\langle a a_{-1} \rangle}{2 \langle a Q a_{-1} \rangle}
\langle a Q \right)
=0.
\end{eqnarray*}
Hence, there exists a scalar $\rho$ such that
\begin{eqnarray*}
&& a_{\tau}\rangle
+\frac{ a_{-1} \rangle}{\langle a Q
a_{-1} \rangle}+\frac{\langle a_1  a\rangle}{2  \langle a_1 Q a
\rangle} Q a\rangle-\frac{  a_1 \rangle }{\langle a_1 Q a
\rangle}-\frac{\langle a a_{-1} \rangle}{2 \langle a Q a_{-1} \rangle}
Q a\rangle=\rho a\rangle;\\
&& \langle a_{\tau} -\frac{ \langle
a_1}{\langle a_1 Q a\rangle} +\frac{\langle a_1  a\rangle}{2  \langle
a_1 Q a \rangle} \langle a Q+\frac{\langle a_{-1}}{ \langle a Q a_{-1}
\rangle}-\frac{\langle a a_{-1} \rangle}{2 \langle a Q a_{-1} \rangle}
\langle a Q=-\rho \langle a .
\end{eqnarray*}
From the above two identities it follows that $\rho=0$ and
\begin{eqnarray*}
 a_{\tau}\rangle
+\frac{a_{-1} \rangle}{\langle a Q
a_{-1} \rangle}+\frac{\langle a_1  a\rangle}{2  \langle a_1 Q a
\rangle} Q a\rangle-\frac{ a_1 \rangle }{\langle a_1 Q a
\rangle}-\frac{\langle a a_{-1} \rangle}{2 \langle a Q a_{-1} \rangle}
Q a\rangle=0.
\end{eqnarray*}
Due to $\langle a= (i, {\boldsymbol u}^T)$,  the above
equation becomes into
\begin{eqnarray*}
 \boldsymbol u_{\tau}
+\frac{ \boldsymbol u_{-1}}{\boldsymbol u^T  \boldsymbol u_{-1} 
+1}+\frac{\boldsymbol u^T_1 
\boldsymbol u  -1}{2 (\boldsymbol u^T_1  \boldsymbol u +1)}  \boldsymbol u 
-\frac{  \boldsymbol u_1 
}{\boldsymbol u^T_1  \boldsymbol u +1}-\frac{\boldsymbol u^T \boldsymbol 
u_{-1}-1}{2 (\boldsymbol u^T 
\boldsymbol u_{-1}+1)}  \boldsymbol u =0.
\end{eqnarray*}
So the unit vector $\boldsymbol u$ satisfies the equation
\begin{eqnarray}\label{equU}
 \boldsymbol u_{\tau}=2(\cS_\nu-1)\frac{ \boldsymbol u_{-1}+\boldsymbol u 
}{|\boldsymbol u^T+  \boldsymbol u_{-1}|^2
 }.
\end{eqnarray}
We sum up the result in the following theorem:
\begin{The}\label{theorM}
The evolutionary vector equation (\ref{equU}) possesses a Lax-Darboux 
representation
(\ref{lins}) with 
\begin{eqnarray}\label{Mu}
&&M_\nu(\boldsymbol u)=I_{n+2}+\frac{ i \nu}{\lambda-i\nu} Q a\rangle
\langle a
-\frac{ i \nu}{\lambda+i\nu} a\rangle \langle a Q, \quad \langle a= (i, {\boldsymbol u}^T), \ \bsu\in S^n; \\
&&U_{\nu}=\frac{i \nu}{\lambda -i \nu} \frac{Q a_{-1} \rangle \langle a- a
\rangle \langle a_{-1} Q}{\langle a Q a_{-1} \rangle} -\frac{i \nu}{\lambda +i
\nu} \frac{ a_{-1} \rangle \langle a Q- Q a
\rangle \langle a_{-1} }{\langle a Q a_{-1} \rangle},\nonumber
\end{eqnarray}
where  $I_{n+2}$ is the identity matrix of size $n+2$,
$Q=\mbox{diag}(-1,1,\ldots,1)$ and $\nu\in \mathbb{R}$ is a constant.
\end{The}
Equation (\ref{equU}) is is known as the discrete Heisenberg spin chain
\cite{ragsan},
whose link with discrete geometry was studied in \cite{bobsu}.

It follows from (\ref{darM}) that  the unit vector
$\boldsymbol u$ is related to the unit vector $\bsv$ by 
\begin{eqnarray}\label{uv}
{\boldsymbol u}=\frac{\boldsymbol v_1+\boldsymbol v}{|\boldsymbol 
v_1+\boldsymbol v|},\qquad \boldsymbol v=\left(\begin{array}{c}
\beta\\\boldsymbol \alpha\end{array}\right) .
\end{eqnarray}
In the following proposition, we give the evolutionary equation for the unit 
vector $\boldsymbol v$.
\begin{Pro}
Equation (\ref{equU}) and equation
\begin{eqnarray}\label{equV}
\boldsymbol v_\tau =\frac{|\boldsymbol v_{-1}+\boldsymbol v|^2 (\boldsymbol 
v_1+\boldsymbol v)-|\boldsymbol v_1+\boldsymbol v|^2
(\boldsymbol v_{-1}+\boldsymbol v)}{( \boldsymbol v^T_1+\boldsymbol v^T)
(\boldsymbol v_{-1}+\boldsymbol v) +|\boldsymbol
v_{-1}+\boldsymbol v||\boldsymbol v_1+\boldsymbol v|}
\end{eqnarray}
are related by the Miura transformation (\ref{uv}).
\end{Pro}
\begin{proof}
Using (\ref{uv}), we can rewrite (\ref{equV}) into
\begin{eqnarray*}
 \boldsymbol v_\tau =\frac{|\boldsymbol v_{-1}+\boldsymbol v|^2 |\boldsymbol 
v_1+\boldsymbol v| \boldsymbol u-|\boldsymbol
v_1+\boldsymbol v|^2
|\boldsymbol v_{-1}+\boldsymbol v| \boldsymbol u_{-1}}{( \boldsymbol u^T
\boldsymbol  u_{-1}+1)|\boldsymbol
v_{-1}+\boldsymbol v||\boldsymbol v_1+\boldsymbol v|}
=\frac{|\boldsymbol v_{-1}+\boldsymbol v| \boldsymbol u-|\boldsymbol
v_1+\boldsymbol v| \boldsymbol u_{-1}}{ \boldsymbol u^T 
\boldsymbol u_{-1}+1} .
\end{eqnarray*}
This leads to 
\begin{eqnarray*}
 \frac{\boldsymbol v_{1\tau}+\boldsymbol v_\tau}{|\boldsymbol v_1+\boldsymbol 
v|}=\frac{
\boldsymbol u_1}{\boldsymbol u^T\boldsymbol u_{1}+1}-\frac{|\boldsymbol
v_2+\boldsymbol v_1| \boldsymbol u}{(\boldsymbol u^T \boldsymbol u_{1} 
+1)|\boldsymbol v_1+\boldsymbol
v|}+\frac{|\boldsymbol v_{-1}+\boldsymbol v| \boldsymbol u}{(
\boldsymbol u^T\boldsymbol u_{-1}
+1)|\boldsymbol v_1+\boldsymbol v|}-\frac{ \boldsymbol u_{-1}}{
\boldsymbol u^T \boldsymbol u_{-1}+1}
\end{eqnarray*}
Substituting this into $\boldsymbol u_\tau$ and using the fact that 
$\boldsymbol u$ is the
unit vector, we have
\begin{eqnarray*}
 \boldsymbol u_\tau=\frac{\boldsymbol v_{1\tau}+\boldsymbol 
v_\tau}{|\boldsymbol v_1+\boldsymbol
v|}- \boldsymbol u^T \frac{\boldsymbol v_{1\tau}+\boldsymbol 
v_\tau}{|\boldsymbol v_1+\boldsymbol v|}
\boldsymbol u=\frac{\boldsymbol u_1+\boldsymbol u}{ \boldsymbol u^T 
\boldsymbol u_{1} +1}-\frac{\boldsymbol
u+  \boldsymbol u_{-1}}{ \boldsymbol u^T \boldsymbol 
u_{-1}+1}=2(\cS_\nu-1)\frac{ \boldsymbol u_{-1}+\boldsymbol u 
}{|\boldsymbol u^T+  \boldsymbol u_{-1}|^2
 }
\end{eqnarray*}
and thus we proved the statement.
\end{proof}
Both equations (\ref{alphaeqq}) and (\ref{equV}) are obtained from
the same Darboux matrix $M_{\nu}$. Thus they share the same
generalised symmetries and conserved densities derived from the zero curvature
conditions \cite{mr86h:58071,mik12, mik13, MPW}. Therefore, equation
(\ref{equV}) can
be viewed
as a symmetry of the nonevolutionary equations (\ref{alphaeqq}) and 
(\ref{backx}).

In 2008, Adler presented the classification of isotropic integrable
Volterra-type 
lattices on the sphere \cite{Adler08a}.
The author gave a list including six integrable equations denoted by V1--V6.
Equation (\ref{equU}) is the V6  when $\delta=1$ in the list.
Equation (\ref{equV}) after scaling of $\tau$ is the V5 when $\epsilon=1$ and 
$k=2$. 
Here we established the relation between these two equations.

As pointed out by Adler \cite{Adler08a}, equation (\ref{equU}) is integrable
without the constraint 
vector $\boldsymbol u$ being the unit vector. Introducing the transformation 
$$\boldsymbol w=\frac{ \boldsymbol u_{-1} +\boldsymbol u}{ |\boldsymbol 
u_{-1}+\boldsymbol u|^2},$$
we have
\begin{eqnarray}\label{equW}
 \boldsymbol w_{\tau}=2( \boldsymbol w^T\boldsymbol w) (\boldsymbol 
w_1-\boldsymbol w_{-1})-4 \boldsymbol w^T (\boldsymbol w_1-\boldsymbol
w_{-1}) \boldsymbol w,
\end{eqnarray}
which is one of vector modified Volterra lattices recently discussed in 
\cite{Adler08b}. 

Equation (\ref{equW}) is well studied in \cite{Adler99}, as an example of 
integrable Jordan triple.  It is a Hamiltonian system with the Hamiltonian 
operator is
\begin{eqnarray*}
 \cH_{\boldsymbol w}=(\boldsymbol w^T \boldsymbol w-2 \boldsymbol 
w {\boldsymbol w}^T) (\cS-\cS^{-1}) (\boldsymbol w^T \boldsymbol 
w-2 \boldsymbol w {\boldsymbol w}^T).
\end{eqnarray*}
and the Hamiltonian function is $\rho_{\boldsymbol w}=-\ln |\boldsymbol w|$. The 
authors also gave its master symmetry 
\begin{eqnarray*}
 \tau_{\boldsymbol w}=(n-1) \boldsymbol w_{\tau}+2 (\boldsymbol w^T \boldsymbol w)
\boldsymbol w_1-4 (\boldsymbol w^T
 \boldsymbol w_1) \boldsymbol w
\end{eqnarray*}
to generate infinitely many commuting generalised symmetries and conserved
densities starting from $\boldsymbol w_\tau$ and $\rho_{\boldsymbol w}$.

Using the Miura transformations, we are able to write down the corresponding 
Hamiltonian operators \cite{kmw13}, Hamiltonian functions and master symmetries
for equations (\ref{equU}) and (\ref{equV})  as follows:
\begin{eqnarray*}
&&{\cal H}_{\boldsymbol u}=(\cS-1) (\cS+1)^{-1}, \quad \rho_{\boldsymbol 
u}=\ln(|\boldsymbol u_1+\boldsymbol u|);\\
&&\tau_{\boldsymbol u}=n \boldsymbol u_{\tau} +\frac{ \boldsymbol u_{1}+\boldsymbol u}{ 
|\boldsymbol u_1+\boldsymbol u|^2};
\\&&\tau_{\boldsymbol v}=n {\bsv}_{\tau}; \qquad \rho_{\boldsymbol
v}=\ln\left(\frac{(\boldsymbol v^T_1+\boldsymbol v^T)(\boldsymbol v_{-1}+\boldsymbol v)}
{|\boldsymbol v_{-1}+\boldsymbol v||\boldsymbol v_1+\boldsymbol v|}+1\right),\label{vmaster}
\end{eqnarray*}
where we use the lower-index $\bsu$ and $ \bsv$ to indicate their 
correspondences to  equations (\ref{equU}) and (\ref{equV}) respectively.
Using the above master symmetry, we are able to compute
higher order symmetries for (\ref{equV}) sharing with both (\ref{alphaeqq}) and 
(\ref{backx}). Recently a new method for construction of master symmetries 
of homogeneous integrable evolution equations (the ${\cal O}$-scheme) was
proposed
in \cite{wang15}. It would be very useful to extend the ${\cal O}$-scheme to
the classes of equations studied in this paper.

For the scalar sine-Gordon equation (\ref{ssg}), the local flow (\ref{equV})
becomes
\begin{eqnarray*}
\theta_{\tau}=\frac{\cos\frac{\theta_{-1}-\theta}{2} \sin\frac{\theta_1+\theta}{2}-\cos\frac{\theta_1-\theta}{2}\sin\frac{\theta_{-1}+\theta}{2}}
{\cos\theta \left(\cos\frac{\theta_1-\theta_{-1}}{2}+1\right)}=\tan \frac{\theta_1-\theta_{-1}}{4} ,
\end{eqnarray*}
which appeared in \cite{NC} as a differential-difference version of the modified Korteweg-de Vries equation.

\section{Yang-Baxter map, integrable partial difference systems and the problem 
of re-factorisation}\label{sec5}

It is well known that the problem of re-factorisation of a product of Darboux 
matrices can be associated with the construction of Yang-Baxter maps 
\cite{vesur, kulpap, sotirmik}. A Darboux (or discrete Lax)
representation for integrable partial difference equations can also be seen as 
the same problem of re-factorisation.

\subsection{The Yang-Baxter map associated with the vector sine-Gordon
equation.}
Let us consider a product of two Darboux matrices and the problem of 
re-factorisation
\begin{equation}\label{refac}
 M_\nu(  \bx)M_\mu(  \by)=M_\mu(  \bY)M_\nu(  
\bX),\qquad \mu\ne\pm\nu,\quad \bx,\by,\bX,\bY\in S^n,
\end{equation}
where  $M_\nu(\bsu)$ is a Darboux matrix (\ref{Mu}) 
defined in Theorem \ref{theorM}.  On the left hand side of (\ref{refac}) the
first and second 
factors are a matrices with poles at $\pm \nu$ and $\pm \mu$ respectively. 
On the right hand side the order of the factors is opposite. If this problem 
of 
a re-factorisation has a unique solution, namely for given $\bx,\by$ one can 
find vectors $\bX,\bY$  uniquely,
then it defines a map $ Y_{\nu,\mu}: S^n\times S^n\mapsto S^n\times S^n$,
\begin{equation}
Y_{\nu,\mu}:(\bx,\by)\mapsto 
\left(\bX(\bx,\by;\nu,\mu),\bY(\bx,\by;\nu,\mu)\right),\qquad 
\bx,\by\in S^n.\label{ybmap} 
\end{equation}
Starting from the map $Y_{\nu,\mu}$ we define maps $Y_{\nu,\mu}^{ij}:S^n\times 
S^n\times 
S^n\mapsto S^n\times S^n\times S^n,\ i<j$
as follows:
\begin{eqnarray*}
\hspace{2.3cm} 
Y_{\nu,\mu}^{12}\,:\,(\bx,\by,\bz)&\mapsto&(\bX(\bx,\by;\nu,\mu),\bY(\bx,\by;\nu
, \mu) ,\bz), \\ 
\hspace{2.3cm} 
Y_{\nu,\kappa}^{13}\,:\,(\bx,\by,\bz)&\mapsto&(\bX(\bx,\bz;\nu,\kappa),\by,
\bY(\bx, \bz;\nu, \kappa)), \\ 
\hspace{2.3cm} 
Y_{\mu,\kappa}^{23}\,:\,(\bx,\by,\bz)&\mapsto&(\bx,\bX(\by,\bz;\mu,\kappa),
\bY(\by, \bz;\mu, \kappa)).
\end{eqnarray*}
\begin{Def}
The map $Y_{\nu,\mu}$ is called a Yang-Baxter map, if it 
satisfies the Yang-Baxter relation
\begin{equation}\label{YB}
 Y_{\nu,\mu}^{12}\circ Y_{\nu,\kappa}^{13}\circ 
Y_{\mu,\kappa}^{23}=Y_{\mu,\kappa}^{23}\circ Y_{\nu,\kappa}^{13}\circ 
Y_{\nu,\mu}^{12}.
\end{equation}
\end{Def}
It is easy to show that the map (\ref{ybmap}) associated with the problem of 
re-factorisation (\ref{refac}) satisfies the Yang-Baxter relation if the 
factorisation of the triple product $M_\nu(  \bx)M_\mu(\by)M_\kappa(  \bz)$ 
is unique \cite{ves04, kulpap}.
\begin{The}\label{YBo}
 Suppose that
 \begin{enumerate}
  \item[{\rm (i)}] the re-factorisation problem (\ref{refac}) defines a unique map $Y_{\nu,\mu}$;
  \item[{\rm (ii)}] for three different real numbers $\mu,\nu, \kappa$ satisfying $\mu\ne-\nu$, $\mu\ne-\kappa$ and $\kappa\ne-\nu$ it follows from
  \[
  M_\nu(  \bx)M_\mu(  
\by)M_\kappa(  \bz)=M_\nu(  \bX)M_\mu(  
\bY)M_\kappa(  \bZ)
 \]
  that 
  $ 
   \bx=\bX,\quad \by=\bY,\quad \bz=\bZ,
  $ that is, the factorisation is unique.
\end{enumerate}
Then the corresponding map $Y_{\nu,\mu}$ is Yang-Baxter.
\end{The}
\begin{proof}
 For any $\bx,\by,\bz\in S^n$ it follows from (\ref{refac}) that
 {\small
 \[
  \begin{array}{ccccccc}
M_\nu(\bx)M_\mu(\by)M_\kappa(\bz)\!\!\!\!&\!\!=\!\!&\!\!\!\!M_\mu(\bY)M_\nu(\bX)
M_\kappa(\bz)\!\!\!\!&\!\!= 
\!\!\!\!&\!\! 
M_\mu(\bY)M_\kappa(\bZ)M_\mu(\tilde{\bX}
)\!\!\!\!&\!\!=\!\!\!\!&\!\!M_\kappa(\tilde{\bZ})M_\mu(\tilde{\bY})M_\nu(\tilde{\bX})\\
(\bx,\by,\bz)&\overset{Y_{\nu,\mu}^{12}}{\longmapsto}&(\bX,\bY,\bz)&
\overset{Y_{\nu,\kappa}^{13}}{\longmapsto}&(\tilde{\bX},\bY,\bZ
)&\overset { Y_ { \mu , \kappa } ^ { 23 } }
{ \longmapsto} &(\tilde{\bX},\tilde{\bY},\tilde{\bZ})
  \end{array}
 \]
}
Swapping the matrices in a different order we get
{\small
 \[
  \begin{array}{ccccccc}
M_\nu(\bx)M_\mu(\by)M_\kappa(\bz)\!\!\!\!&\!\!=\!\!\!\!&\!\!M_\nu(\bx)M_\kappa(\bZ)M_\mu(\bY)\!\!\!\!&\!\!= 
\!\!\!\!&\!\!
M_\kappa(\hat \bZ)M_\nu(\bX)M_\mu(\bY
)\!\!\!\!&\!\!=\!\!\!\!&\!\!M_\kappa(\hat{\bZ})M_\mu(\hat{\bY})M_\nu(\hat{\bX})\\
(\bx,\by,\bz)&\overset{Y_{\mu,\kappa}^{23}}{\longmapsto}&(\bx,\bY,\bZ)&
\overset{Y_{\nu,\kappa}^{13}}{\longmapsto}&({\bX},\bY,\hat \bZ
)&\overset { Y_ { \nu , \mu } ^ { 12 } }
{ \longmapsto} &(\hat{\bX},\hat{\bY},\hat{\bZ})
  \end{array}
 \]
}
Thus we have
\[
 Y_{\mu,\kappa}^{23}\circ Y_{\nu,\kappa}^{13}\circ 
Y_{\nu,\mu}^{12}\, :\,(\bx,\by,\bz)\mapsto 
(\tilde{\bX},\tilde{\bY},\tilde{\bZ}),\quad  Y_{\nu,\mu}^{12}\circ 
Y_{\nu,\kappa}^{13}\circ 
Y_{\mu,\kappa}^{23}\,:\,(\bx,\by,\bz)\mapsto (\hat{\bX},\hat{\bY},\hat{\bZ}).
\]
The uniqueness of the factorisation that 
$(\hat{\bX},\hat{\bY},\hat{\bZ})=
(\tilde{\bX},\tilde{\bY},\tilde{\bZ})$ leads to the required Yang-Baxter relation (\ref{YB})
and thus we complete the proof.
\end{proof}

Let us compute the Yang-Baxter map (\ref{ybmap}) corresponding to the Darboux 
matrix (\ref{Mu}).
 
\begin{The}\label{thYB}
 The Yang-Baxter map (\ref{ybmap}) corresponding to the Darboux matrix 
(\ref{Mu}) is of the form 
\begin{eqnarray}\label{XX}
 \bX(\bx,\by;\nu,\mu)&=&\frac{(\nu^2-\mu^2) \bx+2\mu (\nu+\mu 
(\bx^T\by))\by}{\nu^2+\mu^2+2\mu\nu (\bx^T\by)}   ,\\ \label{YY}
 \bY(\bx,\by;\nu,\mu)&=&\frac{(\mu^2-\nu^2) \by+2\nu (\mu+\nu 
(\bx^T\by))\bx}{\nu^2+\mu^2+2\mu\nu (\bx^T\by)},
\end{eqnarray}
where $\mu\ne\pm\nu$ are constant and $\bx,\by,\bX,\bY\in S^n$.  The map is 
involutive $ Y_{\nu,\mu}\circ Y_{\nu,\mu}=\id$ and 
it has the invariant $\bx^T \by=\bX^T \bY$.
\end{The}
\begin{proof} We derive the map using the  re-factorisation of two Darboux 
matrices (\ref{refac}) and adapting the notations
\[
 a\rangle=\left(\begin{array}{c}
           i\\ \bX
          \end{array}\right)
,\quad b\rangle=\left(\begin{array}{c}
           i\\ \bY
          \end{array}\right),\quad 
\tilde{a}\rangle=\left(\begin{array}{c}
           i\\ \bx
          \end{array}\right),\quad \tilde{b}\rangle=\left(\begin{array}{c}
           i\\ \by
          \end{array}\right) .
\]
At $\lambda=\infty$ equation (\ref{refac}) is satisfied.
Taking 
the residue of  (\ref{refac}) at $\lambda=i\nu$ we get 
\begin{equation}
 \label{rfeq}
Q\tilde{a}\rangle\langle\tilde{a}\left(I+\frac{\mu}{\nu-\mu}Q\tilde{b}
\rangle\langle\tilde{b}-\frac{\mu}{\nu+\mu}\tilde{b}
\rangle\langle\tilde{b}Q\right)=\left(I+\frac{\mu}{\nu-\mu}Q{b}
\rangle\langle{b}-\frac{\mu}{\nu+\mu}{b}
\rangle\langle{b}Q\right)Qa\rangle\langle a ,
\end{equation}
which implies
\begin{eqnarray}\label{eq1}
 &&Q\tilde{a}\rangle= \left(I+\frac{\mu}{\nu-\mu}Q{b}
\rangle\langle{b}-\frac{\mu}{\nu+\mu}{b}
\rangle\langle{b}Q\right)Qa\rangle\Gamma \\ \label{eq2}
&&\langle a=\Gamma \langle\tilde{a}\left(I+\frac{\mu}{\nu-\mu}Q\tilde{b}
\rangle\langle\tilde{b}-\frac{\mu}{\nu+\mu}\tilde{b}
\rangle\langle\tilde{b}Q\right)
\end{eqnarray}
for some scalar $\Gamma\ne 0$.
Taking the first entry in (\ref{eq2}) we find that 
\begin{equation}
 \label{Gammaa}
 \Gamma=\frac{\nu^2-\mu^2}{\nu^2+\mu^2+2\mu\nu (\bx^T\by)}.
\end{equation}
Substitution of (\ref{Gammaa}) in (\ref{eq2}) leads to (\ref{XX}). Similarly, 
 the residue of (\ref{refac}) at $\lambda=i\mu$ yields (\ref{YY}).
 
By direct calculation, we can check that $\bx^T\by=\bX^T\bY$ using the fact $\bx, \by\in S^n$. Indeed,
\begin{eqnarray*}
\bX^T\bY= \frac{((\nu^2-\mu^2)^2+4 \mu\nu (\nu+\mu (\bx^T\by))(\mu+\nu (\bx^T\by)))(\bx^T\by)}
{(\nu^2+\mu^2+2\mu\nu (\bx^T\by))^2} = \bx^T\by .
\end{eqnarray*}
To prove the involutivity of the map, we simply write the map into the matrix form
\begin{eqnarray}\label{pma}
 \left( \bX, \bY\right)=\left(\bx, \by \right)P(\bx,\by),
\end{eqnarray}
where 
\[
P(\bx,\by)=\frac{1}{\nu^2+\mu^2+2\mu\nu (\bx^T\by)}
\left(\begin{array}{cc} (\nu^2-\mu^2)& 2\mu \nu+2\nu^2 
(\bx^T\by)\\
2\nu \mu+2\mu^2 
(\bx^T\by)&\mu^2-\nu^2\end{array} \right)
\]
Notice that $P(\bx,\by)=P(\bX,\bY)$ due to the invariance  $\bx^T\by=\bX^T\bY$.
It can be easily checked that  $P^2(\bx,\by)=I_2$,
which
immediately leads to the involutivity of the map.

It follows from (\ref{pma}) that 
\begin{equation}
 \left( \bx, \by\right)=\left(\bX, \bY
\right)P(\bX,\bY),\label{inP}
\end{equation}
 which compatible to identities (\ref{eq1}) and the one
obtained from the residue
of (\ref{refac}) at $\lambda=i\nu$ and $\lambda=i\mu$, respectively.
Substitution of either (\ref{pma}) or (\ref{inP}) in (\ref{rfeq}) leads to the
identity.

Finally, we show that the factorisation is unique. We know that the Darboux
matrix is invertible and $(M_\nu( \bx))^{-1}=M_\nu^T( \bx)$. Thus if 
$ M_\nu(\bx)M_\mu(\by)M_\kappa(  \bz)=M_\nu(  \bX)M_\mu( \bY)M_\kappa(  \bZ)$,
then
$$M_\mu(\by)M_\kappa(  \bz)=M_\nu^T( \bx)M_\nu(  \bX)M_\mu(  
\bY)M_\kappa(  \bZ).$$
The left hand side is regular at $\lambda=\nu$. The second order pole at
$\lambda=\nu$ in the right hand side vanishes only if 
$\bx^T \bX=1$. Since both $\bx$ and $\bX$ are unit vectors,  we have
$\bx=\bX$. Then the right hand side of the equation is regular at $\lambda=\nu$.
In the same way, we can prove that $\by=\bY$ and $\bz=\bZ$.
According to Theorem \ref{YBo}, the map defined by (\ref{XX}) and (\ref{YY}) is
a Yang-Baxter map.
\end{proof}

\begin{Pro}\label{inv2}
 Yang-Baxter map (\ref{XX}), (\ref{YY}) has a linear vector invariant
 \begin{equation}\label{ybinv}
  \nu \bx+\mu \by=\nu \bX+\mu\bY.
 \end{equation}
\end{Pro}
\begin{proof}
 The identity  can be proved by a direct substitution of (\ref{XX}) 
and (\ref{YY}) in (\ref{ybinv}). Another way to prove this statement is to 
consider the Laurent expansion of equation (\ref{refac}) in $\lambda$ at 
infinity. Then (\ref{ybinv}) follows immediately from the coefficients at 
$\lambda^{-1}$.
\end{proof}

\subsection{Integrable partial difference system}
Let us consider two Darboux transformations for the vector
sine-Gordon system (\ref{vsG}) corresponding to Darboux matrices $M_\nu(\bx)$
and $M_\mu(\bY)$ with distinct positions of the poles $\mu\ne\pm\nu$. According
to Theorem \ref{th1} the vectors $\bX,\by$ can be expressed in terms of the
original variables $\bsv=(\beta,\bsalpha^T)^T$ as
\begin{equation}
 \label{bXby}
 \bX=\frac{\cS_\nu(\bsv)+\bsv}{|\cS_\nu(\bsv)+\bsv|},\qquad
\by=\frac{\cS_\mu(\bsv)+\bsv}{|\cS_\mu(\bsv)+\bsv|},
\end{equation}
where $\cS_\nu$ and $\cS_\mu$ are the corresponding shift automorphisms. 
The Bianchi permutability condition for these two Darboux transformations 
(\ref{dis}) has the form
\begin{equation}\label{bianchi}
 M_\nu(\cS_\mu(\bX))M_\mu(  \by)=M_\mu(\cS_\nu(\by))M_\nu(  
\bX),
\end{equation}
which coincides with (\ref{refac}) where
\begin{equation}
 \bx=\cS_\mu(\bX)=\frac{\bsv_{1,1}+\bsv_{0,1}}{|\bsv_{1,1}+\bsv_{0,1}|},\quad
\bY=\cS_\nu(\by)=\frac{\bsv_{1,1}+\bsv_{1,0}}{|\bsv_{1,1}+\bsv_{1,0}|}.\label{bxbY}
\end{equation}
Here we used the notation $\cS_{\nu}^n\cS_{\mu}^m \bsv=\bsv_{n,m}$ introduced in Section \ref{sec2} 
and convention $\bsv_{0,0}=\bsv$.
Having made  identifications  (\ref{bXby}),(\ref{bxbY}) and using Theorem 
\ref{thYB} we can show that the Bianchi permutability condition (\ref{bianchi}) 
is equivalent to a single quadrilateral equation  for variables 
$\bsv_{i,j}\in S^n$ on  the two dimensional lattice $\Z^2$. This equation can 
be written in a few equivalent forms. 

\begin{The}\label{DIS}
Let $\bsv\in S^n$, $\nu\neq \mu \in \bbbr$ and 
\begin{eqnarray*}
 \bbf=\frac{(\nu^2-\mu^2) \bx+2\mu(\nu+\mu 
(\bx^T\by))\by}{\nu^2+\mu^2+2\mu\nu (\bx^T\by)} ,\quad 
\bF=\frac{(\nu^2-\mu^2) \bX+2\mu(\nu+\mu 
(\bX^T\bY))\bY}{\nu^2+\mu^2+2\mu\nu (\bX^T\bY)}
\\
\bbg=\frac{(\mu^2-\nu^2)\by+2\nu (\mu+\nu 
(\bx^T\by))\bx}{\nu^2+\mu^2+2\mu\nu (\bx^T\by)}
 ,\quad 
 \bG=\frac{(\mu^2-\nu^2)\bY+2\nu (\mu+\nu 
(\bX^T\bY))\bX}{\nu^2+\mu^2+2\mu\nu (\bX^T\bY)}
 \end{eqnarray*}
where 
 \begin{equation}\label{th5e}
\bx=\frac{\bsv_{1,1}+\bsv_{0,1}}{|\bsv_{1,1}+\bsv_{0,1}|},\qquad  
\by=\frac{\bsv_{0,1}+\bsv}{|\bsv_{0,1}+\bsv|}\qquad 
\bX=\frac{\bsv_{1,0}+\bsv}{|\bsv_{1,0}+\bsv|},\qquad  
\bY=\frac{\bsv_{1,1}+\bsv_{1,0}}{|\bsv_{1,1}+\bsv_{1,0}|}.
 \end{equation}
Then the following equations are equivalent
\[\begin{array}{llll}
\mbox{\rm (a)}&\ \ \dfrac{ \bsv_{1,0}+\bsv}{|\bsv_{1,0}+\bsv|}
=\bbf ,\qquad & \mbox{\rm (b)}&\ \ 
 \dfrac{ \bsv_{1,1}+\bsv_{1,0}}{|\bsv_{1,1}+\bsv_{1,0}|}
 =\bbg,\\   \\ 
\mbox{\rm (c)}&\ \  \dfrac{ \bsv_{0,1}+\bsv}{|\bsv_{0,1}+\bsv|}
 =\bG, \qquad &
\mbox{\rm (d)}&\ \ \dfrac{ \bsv_{1,1}+\bsv_{0,1}}{|\bsv_{1,1}+\bsv_{0,1}|}
=\bF, \\  \\
\mbox{\rm (e)}&\ \ \bsv_{1,0}=-\bsv+2\bbf(\bbf^T 
\bsv),  \qquad &
\mbox{\rm (f)}&\ \ \bsv_{1,0}= -\bsv_{1,1}+2\bbg(\bbg^T \bsv_{1,1}),\\  
\nonumber \\
\mbox{\rm (g)}&\ \ \bsv_{0,1}=-\bsv+2\bG(\bG^T \bsv), \qquad &
\mbox{\rm (h)}&\ \ \bsv_{0,1}=-\bsv_{1,1}+2\bF(\bF^T 
\bsv_{1,1})
  \end{array}
 \]
and each of them is equivalent to the Bianchi permutability condition 
(\ref{bianchi}).
\end{The}
\begin{proof}
Let us show that 
equation (a) is equivalent to (e): $\bsv_{1,0}=-\bsv+2\bbf(\bbf^T 
\bsv)$. Indeed, it follows from (a) that 
\[
 2\bbf^T\bsv= 2\frac{ \bsv_{1,0}^T+\bsv^T}{|\bsv_{1,0}+\bsv|}\bsv=2\frac{ 
\bsv_{1,0}^T \bsv+1}{|\bsv_{1,0}+\bsv|}=|\bsv_{1,0}+\bsv|
\]
and in particular $\bbf^T\bsv>0$ and thus (a)$\Rightarrow $(e). Since 
$\bbf^T\bsv>0$ we have  (e)$\Rightarrow $(a):
\[
 \frac{\bsv_{1,0}+\bsv}{|\bsv_{1,0}+\bsv|}=\frac{2\bbf(\bbf^T 
\bsv)}{|2\bbf(\bbf^T \bsv)|}=\bbf.
\]
In a similar way one can show that (b)$\Longleftrightarrow$(f), 
(c)$\Longleftrightarrow$(g) and (d)$\Longleftrightarrow$(h). 
To demonstrate (e)$\Longleftrightarrow$(f) we show that 
\begin{equation}\label{id0}
 -\bsv+2\bbf(\bbf^T \bsv)= -\bsv_{1,1}+2\bbg(\bbg^T \bsv_{1,1})
\end{equation}
by a direct computation. 
Let $|\bsv_{1,1}+\bsv_{0,1}|=x_s$ and $|\bsv_{0,1}+\bsv|=y_s$.
Then 
\begin{equation}\label{id1}
 x_s^2=2 (\bsv_{1,1}^T \bsv_{0,1}+1),\quad y_s^2=2 (\bsv_{0,1}^T 
\bsv+1),\quad 
\bx^T\by=x_s^{-1} y_s^{-1}(\bsv_{1,1}^T \bsv_{0,1}+\bsv_{0,1}^T
\bsv+\bsv_{1,1}^T\bsv+1),
\end{equation}
which implies 
\begin{equation}\label{id2}
 (\bsv_{0,1}^T+\bsv^T)
\bsv_{1,1}=x_s y_s \bx^T\by -\frac{1}{2} 
y_s^2,\quad (\bsv_{1,1}^T+\bsv_{0,1}^T)\bsv=x_s y_s \bx^T\by -\frac{1}{2} 
x_s^2.
\end{equation}
Using (\ref{id1}), (\ref{id2}) we get 
\begin{eqnarray*}
 &&(\nu^2+\mu^2+2\mu\nu (\bx^T\by))^2(\bbf(\bbf^T \bsv)-\bbg(\bbg^T 
\bsv_{1,1}))\\
&&\qquad =\left( (\nu^2-\mu^2)
(\bsv_{1,1}^T+\bsv_{0,1}^T)\bsv x_s^{-1}+\mu(\nu+\mu
(\bx^T\by))y_s\right)
(\nu^2-\mu^2) x_s^{-1} (\bsv_{1,1}+\bsv_{0,1} )\\
&&\qquad- 2\left((\mu^2-\nu^2)(\bsv_{0,1}^T+\bsv^T)
\bsv_{1,1}y_s^{-1}+\nu (\mu+\nu 
(\bx^T\by))x_s\right)
\nu (\mu+\nu 
(\bx^T\by)) x_s^{-1} (\bsv_{1,1}+\bsv_{0,1})
\\
&&\qquad+2 \left( (\nu^2-\mu^2)
(\bsv_{1,1}^T+\bsv_{0,1}^T)\bsv x_s^{-1}+\mu(\nu+\mu
(\bx^T\by))y_s\right)
\mu(\nu+\mu (\bx^T\by))y_s^{-1} (\bsv_{0,1}+\bsv) \\
&&\qquad- \left((\mu^2-\nu^2)(\bsv_{0,1}^T+\bsv^T)
\bsv_{1,1}y_s^{-1}+\nu (\mu+\nu 
(\bx^T\by))x_s\right)
\mu^2-\nu^2) y_s^{-1} (\bsv_{0,1}+\bsv)\\
&&\qquad =\frac{1}{2} (\nu^2+\mu^2+2\mu\nu (\bx^T\by))^2 (\bsv-\bsv_{1,1}).
\end{eqnarray*}
which leads to the required identity (\ref{id0}). Equation  (\ref{refac}) is
equivalent to (\ref{XX}), (\ref{YY}), which after the identification
(\ref{bXby}), (\ref{bxbY}) implies that the Bianchi permutability condition 
(\ref{bianchi}) is equivalent to equations (a) and (b). Thus 
(\ref{bianchi})
$\Longleftrightarrow$(a)$\Longleftrightarrow$(b)$
\Longleftrightarrow$(e)$\Longleftrightarrow$(f). In the same way we can show
that equation (\ref{inP}) 
$\Longleftrightarrow$(c)$\Longleftrightarrow$(d)$\Longleftrightarrow$(g)$
\Longleftrightarrow$(h). Observation that  
(\ref{bianchi})$\Longleftrightarrow$(\ref{pma})$\Longleftrightarrow$(\ref{inP}) 
 completes the proof.
\end{proof}

Thus, the condition (\ref{bianchi}) is equivalent to a quadrilateral equation 
(one take any equation from the list (a)-(h), depending on 
the problem) which is a new vector isotropic integrable system with dependent 
variable $\bsv\in S^n$ on a sphere. It is natural to refer this system as 
{\sl discrete vector sine-Gordon} equation. Let   us take equation (e)
\begin{equation}\label{qua}
 \bsv_{1,0}=-\bsv+2\bbf(\bbf^T 
\bsv)
\end{equation}
as a representative. The Bianchi permutability condition (\ref{bianchi}) for 
two Darboux matrices play the role of a Darboux (or discrete Lax) representation 
for (\ref{qua}).

Using identification (\ref{bXby}), (\ref{bxbY}) we can recast the vector 
 invariant (\ref{ybinv}) of the Yang-Baxter map (Proposition \ref{inv2}) in a 
local conservation law for (\ref{qua}), that is,
\begin{equation}\label{lcl}
 (\cS_\mu-1)\nu\frac{{\boldsymbol
v}_{1,0}+{\boldsymbol v}}{|{\boldsymbol
v}_{1,0}+{\boldsymbol v}|}=(\cS_\nu-1)\mu\frac{{\boldsymbol
v}_{0,1}+{\boldsymbol v}}{|{\boldsymbol
v}_{0,1}+{\boldsymbol v}|}.
\end{equation}
System (\ref{qua}) possess an infinite hierarchy of local
conservation laws, which can be found using formal diagonalisation
\cite{mik12, mik13} of the Darboux matrices or using the master symmetry 
(\ref{vmaster}).
Notice that for  the equation (\ref{qua}) the unit vector $\bsv$ satisfies the 
identity 
\[
 \frac{(\bsv_{1,1}^T+\bsv_{0,1}^T)(\bsv+\bsv_{0,1})}{|\bsv_{1,1 } 
+\bsv_{0,1}||\bsv+\bsv_{0,1}|}=
\frac{(\bsv_{1,1}^T+\bsv_{1,0}^T)(\bsv+\bsv_{1,0})}{|\bsv_{1,1}
+\bsv_{1,0}||\bsv+\bsv_{1,0}|} ,
\]
which follows from the invariant $\bx^T \by=\bX^T\bY$ (Theorem \ref{thYB}).

For the discrete vector sine-Gordon equation there is a well posed initial value 
problem  with initial data
given on the staircase
\[ 
 \{\bsv_{k,k},\bsv_{k,k+1}\,|\, k\in\Z,\ \bsv_{k,k}+\bsv_{k,k+1}\ne 0,\
\ |\bsv_{k,k}|=|\bsv_{k,k+1}|=1\}.
\]
To find the values $\bsv_{i,j}$ below the staircase ($i>j$) one  can use 
equation (e) or (f), for values above the staircase ($j>i+1$) it is convenient 
to use equation (g) or (h) and their shifts (by $\cS_\nu$ and $\cS_\mu$).
Thus the elimination map \cite{mwx1}, which is a useful tool for  study 
symmetries, conservation laws and other structures associated with this 
integrable system can be correctly defined.

The local symmetry (\ref{equV}) of the B\"acklund transformation
(\ref{alphaeqq}) and (\ref{backx}) results in two symmetries
 \begin{eqnarray*}
D_{{\tau_\nu}}{\boldsymbol v}& =&\frac{|\boldsymbol v_{-1,0}+\boldsymbol v|^2 
(\boldsymbol 
v_{1,0}+\boldsymbol v)-|\boldsymbol v_{1,0}+\boldsymbol v|^2
(\boldsymbol v_{-1,0}+\boldsymbol v)}{( \boldsymbol
v_{1,0}^T+\boldsymbol v^T)(
\boldsymbol v_{-1,0}+\boldsymbol v)+|\boldsymbol
v_{-1,0}+\boldsymbol v||\boldsymbol v_{1,0}+\boldsymbol v|},\\
D_{{\tau_\mu}}{\boldsymbol v}& =&\frac{|\boldsymbol v_{0,-1}+\boldsymbol v|^2 
(\boldsymbol 
v_{0,1}+\boldsymbol v)-|\boldsymbol v_{0,1}+\boldsymbol v|^2
(\boldsymbol v_{0,-1}+\boldsymbol v)}{( \boldsymbol
v_{0,1}^T+\boldsymbol v^T)( 
\boldsymbol v_{0,-1}+\boldsymbol v)+|\boldsymbol
v_{0,-1}+\boldsymbol v||\boldsymbol v_{0,1}+\boldsymbol v|}
\end{eqnarray*}
of the discrete vector sine-Gordon equation.

B\"acklund transformations (\ref{alphaeqq}) and (\ref{backx}) give the 
non-local symmetries of the discrete vector sine-Gordon equation as follows:
 \begin{eqnarray}
  &&\begin{array}{rc}
D_t\left(\dfrac{{\boldsymbol v}_{1,0}+{\boldsymbol v}}{|{\boldsymbol
v}_{1,0}+{\boldsymbol v}|}\right)
=\frac{1}{2\nu}({\boldsymbol v}-{\boldsymbol v}_{1,0}); &
D_t\left(\dfrac{{\boldsymbol v}_{0,1}+{\boldsymbol v}}{|{\boldsymbol
v}_{0,1}+{\boldsymbol v}|}\right)
=\frac{1}{2\mu}({\boldsymbol v}-{\boldsymbol v}_{0,1});
                 \end{array}\label{backnmt}\\
&&\begin{array}{rc}
(\cS_\nu-1)\dfrac{D_x\boldsymbol
\alpha}{\beta}=-2 \nu\dfrac{{\boldsymbol
\alpha}_{1,0}+{\boldsymbol \alpha}}{|{\boldsymbol
v}_{1,0}+{\boldsymbol v}|}; &
(\cS_\mu-1)\dfrac{D_x\boldsymbol
\alpha}{\beta}=-2 \mu\dfrac{{\boldsymbol
\alpha}_{0,1}+{\boldsymbol \alpha}}{|{\boldsymbol
v}_{0,1}+{\boldsymbol v}|}.
    \end{array}\label{backnmx}
 \end{eqnarray}

In the case $n=1$, corresponding to the scalar sine-Gordon equation (\ref{ssg}), 
we have
\begin{eqnarray*}
\frac{ \bsv_{1,0}+\bsv}{|\bsv_{1,0}+\bsv|}=\left(\begin{array}{c} \cos \frac{\theta_{1,0}+\theta}{2}\\
\sin \frac{\theta_{1,0}+\theta}{2}\end{array}\right);\qquad
\bx=\left(\begin{array}{c} \cos \frac{\theta_{1,1}+\theta_{0,1}}{2}\\
\sin \frac{\theta_{1,1}+\theta_{0,1}}{2}\end{array}\right),\qquad  
\by=\left(\begin{array}{c} \cos \frac{\theta_{0,1}+\theta}{2}\\
\sin \frac{\theta_{0,1}+\theta}{2}\end{array}\right)  .                                             
\end{eqnarray*}
Substituting them into (\ref{qua}), we get
\begin{eqnarray*}
\left\{\begin{array}{l} \sin \frac{\theta_{1,0}+\theta}{2}=\dfrac{(\nu^2-\mu^2) 
\sin \frac{\theta_{1,1}+\theta_{0,1}}{2}
 +2\mu(\nu+\mu \cos \frac{\theta_{1,1}-\theta}{2})\sin  
\frac{\theta_{0,1}+\theta}{2}}{\nu^2+\mu^2+2\mu\nu \cos 
\frac{\theta_{1,1}-\theta}{2}}\\ \\
\cos \frac{\theta_{1,0}+\theta}{2}=\dfrac{(\nu^2-\mu^2) \cos 
\frac{\theta_{1,1}+\theta_{0,1}}{2}
 +2\mu(\nu+\mu \cos \frac{\theta_{1,1}-\theta}{2})\cos \frac{\theta_{0,1}+\theta}{2}}{\nu^2+\mu^2+2\mu\nu \cos \frac{\theta_{1,1}-\theta}{2}}
 \end{array}\right.
\end{eqnarray*}
Using trigonometric identities, we can rewrite it as
\begin{eqnarray*}
\left\{\!\!\!\!\begin{array}{l} 
\left(\!\nu \cos \frac{\theta_{1,1}+\theta_{0,1}+\theta_{1,0}+\theta}{4}+\mu\cos \frac{\theta_{0,1}+\theta_{1,0}+3\theta-\theta_{1,1}}{4}\!\right) 
\!\!\left(\!\nu \sin \frac{\theta_{1,1}+\theta_{0,1}-\theta_{1,0}-\theta}{4}
 +\mu \sin 
\frac{\theta_{0,1}+\theta-\theta_{1,1}-\theta_{1,0}}{4}\!\right)=0;\\ \\
\left(\!\nu \sin \frac{\theta_{1,1}+\theta_{0,1}+\theta_{1,0}+\theta}{4}+\mu \sin \frac{\theta_{0,1}+\theta_{1,0}+3\theta-\theta_{1,1}}{4}\!\right) 
\!\!\left(\!\nu \sin \frac{\theta_{1,1}+\theta_{0,1}-\theta_{1,0}-\theta}{4}
 +\mu \sin \frac{\theta_{0,1}+\theta-\theta_{1,1}-\theta_{1,0}}{4}\!\right)=0,
 \end{array}\right.
\end{eqnarray*}
Since $\nu\neq \pm \mu$, it follows that
\begin{eqnarray*}
\nu \sin \frac{\theta_{1,1}+\theta_{0,1}-\theta_{1,0}-\theta}{4}= \mu \sin \frac{\theta_{1,1}+\theta_{1,0}-\theta_{0,1}-\theta}{4} ,
\end{eqnarray*}
which is the classical discrete scalar sine-Gordon equation known as the 
Bianchi--Hirota equation.

\section*{Appendix}
Here we give a proof of the lemma which we use in the proof of Proposition
\ref{Aform}.
\begin{Lem}\label{AA}
Let a matrix $A\in {\rm Mat}_{n+2}(\bbbc)$ satisfy
\begin{eqnarray}
&&A A^T=A^TA=0;\label{AAt}\\
&&A^*=QAQ,\label{AA*}
\end{eqnarray}
where $Q=\mbox{diag}(-1,1,\ldots,1)$,
then it can be represented as
\begin{eqnarray*}
 A=h \left(\begin{array}{c}-i \\ {\bf m}\end{array}\right)
\left(\begin{array}{cc}i & {\bf n}^T\end{array}\right),\qquad
0\neq h\in \bbbr,\quad  {\bf m}, {\bf n}\in \bbbr^{n+1},\quad |{\bf m}|= |{\bf
n}|=1 .
\end{eqnarray*}
\end{Lem}
\begin{proof}
It follows from (\ref{AA*}) that the form of matrix $A$ is
\begin{eqnarray}\label{matA}
 A=\left(\begin{array}{cc} h& -i {\bf e}^T\\i {\bf f}&H\end{array}\right),
 \quad h\in \bbbr,\quad {\bf e}, {\bf f}\in \bbbr^{n+1},\quad H\in
{\rm Mat}_{n+1}(\bbbr).
\end{eqnarray}
It follows from (\ref{AAt}) that
\begin{eqnarray}
&& h^2={\bf e}^T {\bf e}={\bf f}^T {\bf f};\label{h}\\
&&H H^T={\bf f}\ {\bf f}^T; \quad H^TH={\bf e}\ {\bf e}^T;\label{H}\\
&& h {\bf f}-H {\bf e}=h {\bf e}-H^T {\bf f}=0.\label{mix}
\end{eqnarray}
If $h=0$ then  $A=0$. Now we assume
that $h\ne0$ and thus ${\bf e\ne 0}$,  ${\bf f\ne 0}$.
From (\ref{H}), we obtain that 
${\rm rank}(H)={\rm rank}(H H^T)={\rm rank} ({\bf f}\, {\bf f}^T)=1 $ and $H$
can be
represented as 
$
 H=\delta {\bf f} {\bf e}^T,\ \delta\in\bbbr.
$
From (\ref{mix}) we deduce that $\delta=h^{-1}$. Thus 
$$ A=h \left(\begin{array}{cc} 1& -i h^{-1}{\bf e}^T\\i h^{-1} {\bf
f}&h^{-2}{\bf f e}^T\end{array}\right)=
h \left(\begin{array}{c}-i \\ {\bf m}\end{array}\right)
\left(\begin{array}{cc}i & {\bf n}^T\end{array}\right)
$$
where ${\bf n}={\bf e}/h,\ {\bf m}={\bf f}/h$ and $|{\bf n}|=|{\bf m}|=1$ due to
(\ref{h}).
\end{proof}

\section*{Acknowledgements}
The paper is supported by the Leverhulme Trust, EPSRC grant EP/I038675/1
and  EPSRC grant EP/I038659/1. 
All authors gratefully acknowledge the financial support. 

\end{document}